\documentclass[a4paper,USenglish,cleveref, thm-restate,numberwithinsect]{lipics-v2019}
\title{Expected Linear Round Synchronization: The Missing Link for Linear Byzantine SMR \\ \smaller[2](Extended Version)} %TODO Please add

\titlerunning{Expected Linear Round Synchronization: The Missing Link for Linear Byzantine SMR} %TODO optional, please use if title is longer than one line

\author{Oded Naor}{Technion}{}{}{}

\author{Idit Keidar}{Technion}{}{}{}

\authorrunning{O. Naor and I. Keidar} 

\Copyright{Oded Naor and Idit Keidar} 


 %TODO mandatory; please add comma-separated list of keywords

\nolinenumbers %uncomment to disable line numbering

\hideLIPIcs  %uncomment to remove references to LIPIcs series (logo, DOI, ...), e.g. when preparing a pre-final version to be uploaded to arXiv or another public repository

\EventEditors{Hagit Attiya}
\EventNoEds{1}
\EventLongTitle{34rd International Symposium on Distributed Computing (DISC 2020)}
\EventShortTitle{DISC 2020}
\EventAcronym{DISC}
\EventYear{2020}
\EventDate{October 12--18, 2020}
\EventLocation{Virtual Conference}
\EventLogo{}
\SeriesVolume{179}
\ArticleNo{24}
\usepackage{epsfig,endnotes}
\usepackage{url}
\makeatletter
\g@addto@macro{\UrlBreaks}{\UrlOrds}
\makeatother
\usepackage{xifthen}
\usepackage{ amssymb }
\usepackage{eucal}
\usepackage{multicol}
\usepackage{amsthm}
\usepackage{minibox}
\usepackage{thm-restate}
\usepackage{framed}
\usepackage{fixltx2e}
\usepackage{tabularx}
\usepackage{marginnote}
\usepackage{mathtools}
\usepackage{xcolor}
\usepackage{multirow}
\usepackage{relsize}
\usepackage{bm}
\usepackage{setspace}
\usepackage{nameref}
\usepackage{hyperref}
\usepackage{cleveref}
\usepackage{tikz}

\newcolumntype{L}{>{\raggedright\arraybackslash}X}

\DeclareMathAlphabet{\mathcal}{OMS}{cmsy}{m}{n}

\newtheorem{property}{Property}
\newtheorem{observation}{Observation}
\crefalias{AlgoLine}{line}%
\makeatletter
\let\cref@old@stepcounter\stepcounter
\def\stepcounter#1{%
	\cref@old@stepcounter{#1}%
	\cref@constructprefix{#1}{\cref@result}%
	\@ifundefined{cref@#1@alias}%
	{\def\@tempa{#1}}%
	{\def\@tempa{\csname cref@#1@alias\endcsname}}%
	\protected@edef\cref@currentlabel{%
		[\@tempa][\arabic{#1}][\cref@result]%
		\csname p@#1\endcsname\csname the#1\endcsname}}

\makeatother

\makeatletter
\newcommand*{\textoverline}[1]{$\overline{\hbox{#1}}\m@th$}
\makeatother

\crefformat{chapter}{\S#2#1#3}
\crefmultiformat{chapter}{\S\S#2#1#3}{and~#2#1#3}{, #2#1#3}{, and~#2#1#3}

\crefformat{section}{\S#2#1#3}
\crefmultiformat{section}{\S\S#2#1#3}{and~#2#1#3}{, #2#1#3}{, and~#2#1#3}
\crefname{algocf}{alg.}{algs.}
\Crefname{algocf}{Alg.}{Alg.}
\crefname{equation}{eq.}{eqs.}
\Crefname{equation}{Eq.}{Eqs.}
\crefname{table}{table}{tables}
\Crefname{table}{Table}{Table}
\crefname{definition}{def.}{defs.}
\Crefname{definition}{Def.}{Defs.}

\usepackage[linesnumbered,ruled,nofillcomment]{algorithm2e}
\SetKwBlock{KwFunction}{function}{}
\SetKwBlock{KwUpon}{upon}{}
\SetKwBlock{KwOn}{on}{}
\SetKwBlock{KwAfter}{after}{}
\SetKwBlock{KwInitialize}{initialize}{}
\SetKwBlock{KwTimer}{Timer:}{}
\SetKwBlock{KwLeader}{Leader:}{}
\SetKwBlock{KwState}{initialize}{}
\SetKwBlock{KwRepeat}{repeat}{}
\SetKwBlock{KwPeriodically}{periodically}{}
\SetKwBlock{KwState}{state}{}
\SetKwBlock{KwAtomic}{atomically}{}
\SetKwBlock{KwConcurrently}{concurrently}{}
\SetCommentSty{textup} 
\SetKwComment{KwIttayComment}{(}{)} 
\SetKw{KwAnd}{and} 
\SetKw{KwOr}{or}
\SetKw{KwSetTimer}{setTimer}
\SetKwBlock{KwExpTimer}{on timer expiration}{}
\SetKwFor{ForAll}{forall}{}

\newcommand{\eg}{e.g.\xspace}

\newcommand{\ie}{i.e.\xspace}

\newcommand{\curr}{\ensuremath{\textit{curr\_round}}\xspace}

\newcommand{\dissTime}{\ensuremath{\delta}\xspace}
\newcommand{\commentx}[1]{/* #1 */}

\newcommand{\nextVar}{\ensuremath{\textit{next\_round}}\xspace}

\newcommand{\send}{\ensuremath{\textbf{send}}\xspace} 
\newcommand{\broadcast}{\ensuremath{\textbf{broadcast}}\xspace} 
 
\newcommand{\relay}[1]{\ensuremath{\textsc{Relay}(#1)}\xspace}

\newcommand{\wishToAdvance}{\ensuremath{\textsf{advance}()}\xspace}
\newcommand{\wishToAdvanceSubscript}[1]{\ensuremath{\textsf{advance}_{#1}()}\xspace}
\newcommand{\proposeRound}[1]{\ensuremath{{\textsf{ new\_round}}(#1)}\xspace} 
\newcommand{\proposeRoundSubscript}[2]{\ensuremath{{\textsf{ new\_round}}_{#2}(#1)}\xspace}

\newcommand{\newLeader}[1]{\ensuremath{{\textsf{round\_leader}}(#1)}\xspace}
\newcommand{\newLeaderSubscript}[2]{\ensuremath{{\textsf{round\_leader}_{#2}}(#1)}\xspace}  

\newcommand{\newRound}[1]{\ensuremath{\left\langle\textsf{pre-commit},#1\right\rangle}\xspace}

\newcommand{\finalize}[1]{\ensuremath{\left\langle\textsf{finalize},#1\right\rangle}\xspace} 
\newcommand{\finalizeAck}[1]{\ensuremath{\left\langle \overline{\textsf{finalize}} ,#1\right\rangle}\xspace} 

\newcommand{\wish}{\newRound}
\newcommand{\vote}[1]{\ensuremath{\left\langle \textsf{commit},#1 \right\rangle} \xspace} 
\newcommand{\TC}[1]{%
	\ifthenelse{\isempty{#1}}%
	{\ensuremath{\left\langle \overline{\textsf{pre-commit}}\right\rangle}\xspace}
	{\ensuremath{\left\langle \overline{\textsf{pre-commit}},#1\right\rangle}\xspace}
}
\newcommand{\QC}[1]{%
	\ifthenelse{\isempty{#1}}%
	{\ensuremath{\left\langle \overline{ \textsf{commit}} \right\rangle}\xspace}
	{\ensuremath{\left\langle \overline{\textsf{commit}},#1\right\rangle}\xspace}
}

\newcommand{\GST}{\ensuremath{\text{GST}}\xspace}

\newcommand{\process}[1]{\ensuremath{\mathcal{P}_{#1}}\xspace}

\newcommand{\finalized}{\ensuremath{\textit{finalized}}\xspace}
\SetKwProg{myproc}{Procedure}{}{}

\newcommand{\tempRelay}{\ensuremath{\textit{round\_relay}}\xspace}
\newcommand{\true}{\ensuremath{\texttt{True}}\xspace}
\newcommand{\false}{\ensuremath{\texttt{False}}\xspace}

\newcommand{\rMax}[1]{\ensuremath{\textit{r\_max}(#1)}\xspace}
\newcommand{\syncTime}[1]{\ensuremath{\textit{sync\_time}(#1)}\xspace}

\newcommand{\leader}[1]{\ensuremath{\textit{Leader}(#1)}\xspace}

\makeatletter
\let\orgdescriptionlabel\descriptionlabel
\renewcommand*{\descriptionlabel}[1]{%
	\let\orglabel\label
	\let\label\@gobble
	\phantomsection
	\edef\@currentlabel{#1}%
	\let\label\orglabel
	\orgdescriptionlabel{#1}%
}
\makeatother

\makeatletter

\newcommand*\circled[1]{\tikz[baseline=(char.base)]{
		\node[shape=circle,draw=black!80, line width=0.3mm, inner sep=2pt] (char) {\textcolor{black}{#1}};}}

\hypersetup{
	colorlinks=true, % Defaults to red
	citecolor=black,
	linkcolor=black
}

\begin{document}
\maketitle
\begin{abstract}
State Machine Replication (SMR) solutions often divide time into rounds, with a designated leader driving decisions in each round.
Progress is guaranteed once all correct processes \emph{synchronize} to the same round, and the leader of that round is correct.
Recently suggested Byzantine SMR solutions such as HotStuff, Tendermint, and LibraBFT achieve progress with a linear message complexity and a constant time complexity once such round synchronization occurs.
But round synchronization itself incurs an additional cost.
By Dolev and Reischuk's lower bound, any deterministic solution must have $\Omega(n^2)$ communication complexity. Yet the question of randomized round synchronization with an expected linear message complexity remained open.

We present an algorithm that, for the first time, achieves round synchronization with expected linear message complexity and expected constant latency. Existing protocols can use our round synchronization algorithm to solve Byzantine SMR with the same asymptotic performance.
\end{abstract}

\sloppy
\section{Introduction}
Byzantine \emph{State Machine Replication (SMR)} has received a lot of attention in recent years due to the increasing demand for robust and scalable systems. 
In order to tolerate periods of high load or even denial-of-service attacks, practical solutions commonly assume the \emph{eventual synchrony} model~\cite{dwork1988consensus}, meaning that they guarantee consistency despite asynchrony and make progress during periods when the network is synchronous.
Examples of such systems include PBFT~\cite{castro1999practical}, SBFT~\cite{gueta2019sbft}, LibraBFT~\cite{baudet2019librabft}, HotStuff~\cite{yin2019hotstuff}, Zyzzyva~\cite{kotla2007zyzzyva}, Tendermint~\cite{buchman2018tendermint}, and many more.
Eventually synchronous SMR solutions typically iterate through a sequence of \emph{rounds}, (also called views), wherein a designated \emph{leader} process tries to drive all correct processes to consensus.
The main complexity of such algorithms arises whenever a new round begins and its (new) leader collects information about possible consensus decisions in previous rounds.

When using such protocols, it is common to constantly advance in rounds with a rotating leader~\cite{castro1999practical,yin2019hotstuff,baudet2019librabft}; this is because when the leader is faulty, it is possible for some processes to perceive progress while others made no progress.

In the last couple of years, there has been a race to improve the performance of Byzantine SMR.
Recent algorithms such as Tendermint~\cite{buchman2018tendermint}, Casper~\cite{buterin2017casper}, HotStuff~\cite{yin2019hotstuff}, and LibraBFT~\cite{baudet2019librabft} allow rounds to advance (and leaders to be replaced) with a constant time complexity and a linear message complexity.
Thus, even if every consensus instance is led by a different leader, the message complexity for each decision remains linear. 
Nevertheless, the linear message complexity is achieved only \emph{after} all correct processes synchronize to execute the same round of the protocol (provided that that round’s leader is correct).
And such \emph{round synchronization} has a cost of its own.
In Tendermint, round advancement is gossiped throughout the system, entailing an expected $O(n\log n)$ message complexity with expected $O(\log n)$ latency.
In HotStuff, it is delegated to a separate round synchronization module called \emph{PaceMaker}, whose implementation is left unspecified.
And in LibraBFT, this module is implemented with quadratic message complexity, which in fact matches Dolev and Reischuk's $\Omega(n^2)$ communication complexity lower bound~\cite{dolev1985bounds} on deterministic Byzantine consensus.
Later work on Cogsworth~\cite{naor2019cogsworth} implemented a randomized PaceMaker with expected constant latency and linear message complexity under benign failures, but with expected quadratic message complexity in the Byzantine case.

In this work we present a new round synchronization algorithm that achieves expected constant time complexity and expected linear communication complexity even in the presence of Byzantine processes.
Specifically, under an oblivious adversary, we guarantee these bounds on the expected time/message cost until all processes synchronize to the same round from an \emph{arbitrary} state of the protocol.
Under a strong adversary, we achieve the same bounds but on the \emph{average} expected time and message cost until round synchronization over all states occurring in an infinite run of the protocol.
To this end, we decompose the round synchronization module into a \emph{synchronizer} abstraction and two local functions.
The synchronizer abstraction captures the essence of the distributed coordination required in order to synchronize processes to the same round.

Like previous works~\cite{naor2019cogsworth,yin2019hotstuff,baudet2019librabft}, the main technique used in our algorithm to lower the message complexity is a relay-based message distribution with threshold signatures.
Instead of broadcasting messages all-to-all, our algorithm sends each message to a designated \emph{relay}.
The relay aggregates messages from multiple processes, and when a certain threshold is met, it combines them into a threshold signature, which it sends to all the processes.
Note that a threshold signature’s size is the same as the original signature sizes, \ie, remains constant as the number of processes grows.
This leads to linear communication complexity per message.
The challenge is that the relay can be Byzantine and, for example, send the aggregated message only to a subset of the processes.
Another challenge arises when some correct process advances to a new round while others lag behind.
We introduce a relay-based linear-complexity helping mechanism to allow lagging processes to catch up with faster ones without all-to-all broadcast.

In summary, the main contribution of this paper is providing an algorithm that for the first time reduces the expected message complexity of Byzantine SMR in the presence of Byzantine faults to linear, while maintaining expected constant latency.
The rest of the paper is structured as follows: \Cref{sec:model} describes the model; \Cref{sec:problemDef} formally defines the round synchronization problem and our performance metrics; \cref{sec:decomposition} explains our decomposition of round synchronization into a synchronizer abstraction and local functions, and proves that this decomposition solves the round synchronization problem; \Cref{sec:algorithm} presents our new synchronizer algorithm and proves its expected linear message complexity, expected constant latency, and correctness; \Cref{sec:relatedWork} gives related work and \Cref{sec:conclusion} concludes the paper.

\section{Model}
\label{sec:model}

Our model consists of a set $\Pi = \left\lbrace \process{1}, \process{2}, \ldots, \process{n} \right\rbrace$ of $n$ processes.
Every two processes in $\Pi$ have a bidirectional, reliable, and authenticated link between them, \ie, every process can send a message to another process that will eventually arrive and the recipient can verify the sender's identity.
We use the term \emph{broadcast} to indicate sending a message to all processes.

We follow the eventually synchronous model~\cite{dwork1988consensus} in which there is no global clock, and every execution is divided into two periods: first, an unbounded period of asynchrony; and then, a period of synchrony, where messages arrive within a bounded time,~\dissTime.
The second period begins at a moment called the \emph{Global Stabilization Time~(\GST)}.
Messages sent before \GST arrive by $\GST + \delta$.
We assume that after \GST, processes can correctly estimate such an upper bound $\delta$ on communication latency and also measure time locally, but this does not imply a global clock.
We consider a failure model where $f < n/3$ processes may be \emph{faulty}, or \emph{Byzantine} and act arbitrarily.

We assume a shared source of randomness, $\mathcal{R}$, that is used to derive a function $\relay{r,k} \colon \mathbb{N} \times \left\{ 1, \ldots, f+1 \right\} \mapsto \Pi$.
This function is used to select for each round~$r$ the $k$-th process that will act as a relay.
The relay function satisfies the following properties: 
\begin{description}
	\item[{R1}\label{prop:relayFunc:differentProcesses}] $f+1$ different relays for each round:
	\begin{equation*}
	\forall r \in \mathbb{N}, \forall 1 \leq i < j \leq f+1 \colon \relay{r, i} \neq \relay{r, j}.
	\end{equation*}
	\item[{R2}\label{prop:relayFunc:random}] Random relay selection, while ensuring $f+1$ different relays for each round:
	\begin{multline*}
	\forall r \in \mathbb{N}, \forall 1 \leq k \leq f+1, \forall \process{1}, \process{2} \in  \Pi \setminus \bigcup_{i=1}^{k-1} \left\{ \relay{r, i} \right\} \colon \\
	{\Pr \left[ \relay{r,k} = \process{1} \right] = \Pr \left[ \relay{r,k} = \process{2} \right]}.
	\end{multline*}
\end{description} 

Note that \ref{prop:relayFunc:random} implies that the first relay is continuously rotated throughout the run, \ie, $\forall r \in \mathbb{N} \colon \bigcup_{i=r}^{\infty} \relay{i,1} = \Pi$.
Generating secure randomness as assumed by our protocol has been studied in the literature, \eg,~\cite{blum1983coin,cleve1986limits,moran2009optimally,andrychowicz2014secure}, and is beyond the scope of this paper.

For clarity of the algorithm’s presentation, we assume that the adversary is a \emph{static oblivious adversary}~\cite{borodin1992optimal,ben1994power,fiat2002competitive}, \ie, has no knowledge of the randomness~$\mathcal{R}$.
This assumption is required for the worst-case performance bounds as defined in~\Cref{sec:problemDef:performance}, and can be relaxed to a \emph{strong static adversary} if we only wish to prove an average-case bound.
We discuss this in~\Cref{sec:algs:relaxedModel} below. 

Like previous linear-complexity BFT algorithms~\cite{baudet2019librabft,buchman2018tendermint,naor2019cogsworth,yin2019hotstuff}, we use a cryptographic signing scheme, a public key infrastructure (PKI) to validate signatures, and a threshold signing scheme~\cite{boneh2001short, cachin2005random, shoup2000practical}.
The threshold signing scheme is used in order to create a compact-sized signature of $K$-of-$N$ processes as in other consensus protocols~\cite{cachin2005random}.
Usually $K = f+1$ or $K = 2f+1$.
The size of a threshold signature is constant and does depend on $K$ or $N$.
We assume that the adversary is polynomial-time bounded, \ie, the probability that it will break the cryptographic assumptions in this paper (\eg, the cryptographic signatures, threshold signatures, etc.) is negligible.
\section{Problem Definition - Round Synchronization}
\label{sec:problemDef}

We start by specifying the round synchronization problem in \Cref{sec:problemDef:specification}, then discuss performance metrics in~\Cref{sec:problemDef:performance}, and conclude by describing how to use a round synchronization module to solve consensus in~\Cref{sec:problemDef:solveConsensus}.

\subsection{Specification} \label{sec:problemDef:specification}
We define a long-lived task of \emph{round synchronization}, parameterized by the desired round duration $\Delta$.
It has a single output signal at process \process{i}, \newLeaderSubscript{r, \process{}}{i}, $r \in \mathbb{N}, \process{} \in \Pi$, indicating to \process{i} to enter round $r$ of which \process{} is the leader.
We say that a process \process{i} is \emph{in round~$r$} between the time $t$ when \newLeaderSubscript{r, \cdot}{i} occurs and the next \newLeaderSubscript{r',\cdot}{i} event after $t$.
If no such event occurs, \process{i} remains in round $r$ from $t$ onward.
The goal of round synchronization is to reach a \emph{synchronization time}, defined as follows:
\begin{definition} [Synchronization time] \label{def:syncTime}
Time $t_s$ is a \emph{synchronization time} if all correct processes are in the same round $r$ from $t_s$ to at least $t_s + \Delta$, and $r$ has a correct leader.
\end{definition}

A round synchronization module satisfies two properties.
The first ensures that in every round all the correct processes have the same leader.
\begin{property} [Leader agreement] \label{prop:leaderAgreement}
For any two correct processes $\process{i}, \process{i'}$ if \newLeaderSubscript{r,\process{j}}{i} and \newLeaderSubscript{r,\process{j'}}{i'} occur, then $\process{j} = \process{j'}$.
\end{property}

The second property ensures that synchronization times eventually occur.
Formally:

\begin{property} [Eventual round synchronization] \label{prop:roundClockEventualSync}
For every time $t$ in a run, there exists a synchronization time after $t$.
\end{property}

\newcommand{\tRoundSync}[1]{\ensuremath{t_{#1}}\xspace}

\subsection{Performance Metrics} \label{sec:problemDef:performance}
For an oblivious adversary, we measure the maximum expected performance after GST under all possible adversary behaviors and protocol states, where the expectation is taken over random outputs of our randomness source $\mathcal{R}$, which drives the relay function.
In more detail, let $S$ be the set of all reachable states of a round synchronization algorithm, and let $\mathcal{A}$ be the set of all possible adversary behaviors after \GST.
This includes selecting up to $f$ processes to corrupt and scheduling all message deliveries within at most $\dissTime$ time.
For a state $s \in S$, and adversary behavior $a \in \mathcal{A}$, let $\textit{RS}(s,a,\pi)$ be the time from when $s$ occurs until the next synchronization time in a run extending $s$ with adversary behavior $a$ and the relay function derived from the random bits $\pi \in \mathcal{R}$.

The \emph{worst-case expected latency} of the round synchronization module is defined as 
\begin{equation*}
\max_{\substack{s \in S \\ a \in \mathcal{A}}} \left\{ \mathop{{}\mathbb{E}}_{\pi \in \mathcal{R}} \left[ \textit{RS} \left(s,a,\pi \right) \right] \right\}.
\end{equation*}

Similarly, to define message complexity let $M(s,a,r)$ be the total number of messages correct processes send from state $s$ until the next synchronization time in a run extending $s$ with adversary $a \in \mathcal{A}$ and relay output $\pi \in \mathcal{R}$.
The \emph{worst-case message complexity} is defined as 
\begin{equation*}
\max_{\substack{s \in S \\ a \in \mathcal{A}}} \left\{ \mathop{{}\mathbb{E}}_{\pi \in \mathcal{R}} \left[ M \left(s, a, \pi \right) \right] \right\}.
\end{equation*}

For brevity, in the rest of this paper, we omit the parameters $s$, $a$, $\pi$, and simply bound the expected latency or message cost over all reachable states and adversary behaviors.

\subsection{Using Round Synchronization to Solve Consensus} \label{sec:problemDef:solveConsensus}
In HotStuff~\cite{yin2019hotstuff}, Theorem 4 states the following in regards to reaching a decision in the consensus protocol:
\begin{quotation}
``After \GST, there exists a bounded time period $T_f$ such that if all correct replicas remain in view $v$ during $T_f$ and the leader for view $v$ is correct, then a decision is reached.''
\end{quotation}
The round synchronization module satisfies exactly the conditions of the theorem, \ie, an eventual round that all the correct processes are in at the same time for at least $\Delta = T_f$, and the leader of that round is correct.

Given a round synchronization module with expected linear message complexity and expected constant latency, HotStuff solves consensus in the same expected asymptotic message complexity and latency as the round synchronization module.
In addition, HotStuff also uses the same cryptographic primitives (namely threshold signatures) as we use in this paper, incurring similar computational costs.

Note that, in general, processes know neither whether their leader is correct nor whether all correct processes are in the same round as them.
Indeed, it is possible for a set of $f+1$ correct processes (and $f$ Byzantine ones) to make progress in a round with a Byzantine leader, while $f$ correct processes are stuck behind.
In an SMR algorithm where the processes communicate only with the leader of each round and do not broadcast decisions to all processes, this scenario is indistinguishable from one where the leader is correct and all correct processes make progress.
Therefore, to ensure the condition required by HotStuff (and captured by \Cref{prop:roundClockEventualSync}), we continuously advance in rounds and change leaders, regardless of the observed progress made in the consensus protocol utilizing the round synchronization module.

\section{Round Synchronization Decomposition}
\label{sec:decomposition}

We build the round synchronization module using a \emph{synchronizer} abstraction and two local modules.
The synchronizer captures the necessary distributed coordination among the processes.
The abstraction's properties appear in \Cref{sec:problemDef:synchronizer}, and a round synchronization module using this abstraction is given in \Cref{sec:problemDef:localFuncs}.
The latter consists of a \emph{timer} function that paces the synchronizer and a  \emph{leader} function that outputs the leader and round to the application.
This decomposition is illustrated in \Cref{fig:block_diagram}.

\vspace{-5pt}
\subsection{Synchronizer}
\label{sec:problemDef:synchronizer}

We define a \emph{synchronizer} abstraction to be a long-lived task with an API that includes an \emph{\wishToAdvanceSubscript{i}} input and
a \emph{\proposeRoundSubscript{r}{i}} output signal, where $r \in \mathbb{N}$.

In a similar way to the round synchronization module, we say that process \process{i} \emph{enters round~$r$} when \proposeRoundSubscript{r}{i} occurs.
We say \emph{process \process{} is in round~$r$} during the time interval that starts when \process{} enters round~$r$ and ends when it next enters another round.
If the process does not enter a new round, then it remains indefinitely in $r$.
We denote by \rMax{t} the maximum round a correct process is in at time~$t$.

We define four properties a synchronizer algorithm should guarantee.
The first ensures that rounds are monotonically increasing. Formally:
\vspace{-5pt}
\begin{property} [Monotonically increasing rounds] \label{prop:monotonicRounds}
	For each correct process \process{i}, if $\proposeRoundSubscript{r'}{i}$ occurs after  $\proposeRoundSubscript{r}{i}$, then  $r' > r$.
\end{property}
The next property is the validity of new rounds.
\vspace{-8pt}
\begin{property} [Validity] \label{prop:validity}
	If a correct process signals \proposeRound{r} then some correct process called \wishToAdvance while in round~${r-1}$.
\end{property}

Next, we define the two liveness properties.
Informally, the first ensures the stabilization of at least $f+1$ correct processes to the same maximum round, and the second ensures progress after the stabilization.

\begin{property} [Stabilization] \label{prop:synchronizerSyncronization}
	For any $t$ during the run, let $t_0$ be the first time when a correct process enters round $\rMax{t}$. 
	If no correct process enters any round $r > \rMax{t}$, then:
	\begin{description}
		\item[{S1}\label{prop:sync1}] From some time $t_1$ onward, at least $f+1$ correct processes are in round \rMax{t}.
		\item[{S2}\label{prop:sync2}] If $t_0 \geq \GST$ and \relay{\rMax{t},1} is correct, then from some time $t_2$ onward all the correct processes enter \rMax{t} and $t_2-t_0 \leq c_1$ for some constant $c_1$.  
	\end{description}
\end{property}
Although \Cref{prop:synchronizerSyncronization} is primed on no correct processes \emph{ever} entering rounds higher than \rMax{t}, we observe that \ref{prop:sync2} holds as long as no process enters rounds higher than \rMax{t} by $t_0+c_1$ because any such run is indistinguishable to all processes until time $t_0+c_1$ from a run where they never enter a higher round at all.
Formally:
\begin{observation} \label{observ:stabilization}
Assume \Cref{prop:synchronizerSyncronization} holds, then for any $t$ during the run, let $t_0 \geq \GST$ be the first time when a correct process enters round $\rMax{t}$. 
If no correct process enters any round $r > \rMax{t}$ by $t_0+c_1$ for some constant $c_1$ and \relay{\rMax{t},1} is correct, then all correct processes enter round \rMax{t} by $t_0+c_1$.
\end{observation}
The next property ensures progress.
\begin{property} [Progress] \label{prop:synchronizerProgress}
	For any $t$ during the run, if $f+1$ correct processes in round \rMax{t} call \wishToAdvance by $t_0$, and no correct process calls \wishToAdvance while in any round $r > \rMax{t}$ then:
	\begin{description}
		\item[{P1}\label{prop:prog1}] From some time $t_1$ onward, there is at least one correct process in $\rMax{t} + 1$.
		\item[{P2}\label{prop:prog2}] If $t_0 \geq \GST$ and \relay{\rMax{t},1} is correct, then from some time $t_2$ onward all the correct processes enter $\rMax{t}+1$ and  $t_2 - t_0 \leq c_2$ for some constant $c_2$.
	\end{description}
\end{property}
Property \ref{prop:prog2} is not required for round synchronization, but it gives a bound on performance.

\subsection{From Synchronizer to Round Synchronization}
\label{sec:problemDef:localFuncs}

\begin{figure}[t]
	\centering
	\includegraphics[width=0.9\textwidth]{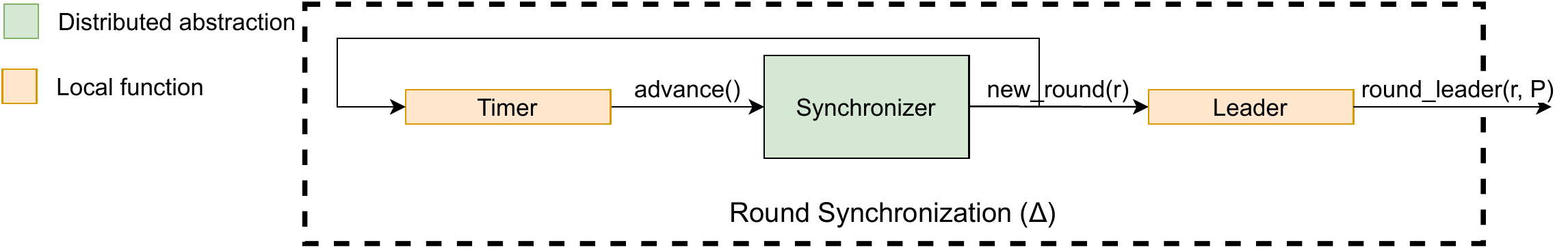}
	\caption{Round synchronization using the synchronizer abstraction.}
	\label{fig:block_diagram}
\end{figure}

We now describe how to use the synchronizer abstraction to implement round synchronization.
The implementation uses two local functions: a \emph{timer} function that paces a process' \wishToAdvance calls, and a \emph{leader} function that maps a round to a leader using the Relay function.
This construction is illustrated in \Cref{fig:block_diagram}, and specified in \Cref{alg:timerLeader}.
When one module invokes a function in another, we refer to this as a \emph{signal}, \eg, the timer signals \wishToAdvance to the synchronizer.

\begin{algorithm*}[t] \small
	\caption{Round synchronization using the synchronizer abstraction.}
	\label{alg:timerLeader}
	\SetAlgoNoEnd
	\DontPrintSemicolon
	
	\KwTimer()
	{
		\KwAfter($c_1 + \Delta$ from last \proposeRound{r} signal: \tcp*[h]{$c_1$ is defined in \ref{prop:sync2}} \label{alg:timerLeader:advance}) 
		{
			\wishToAdvance \;
		}
	}
	\BlankLine
	\KwLeader()
	{
		\KwOn(\proposeRound{r} signal\unskip:)
		{
			\newLeader{r, \relay{r,1}} \;
		}	
	}
\end{algorithm*}

We prove that this construction provides round synchronization.
Let $t^0 = \GST$ and $\forall \ell \geq 1$ let $t^\ell$ be the first time after $t^{\ell-1}$ that a correct process enters a new maximum round.
We prove the following lemma:

\begin{lemma} \label{lem:indefinitelyNewRounds}
	In an infinite run of \Cref{alg:timerLeader}, $t^\ell$ eventually occurs for any $\ell \geq 0$.
\end{lemma}

\begin{proof}
We prove this by induction on $\ell$. 
Based on the model, the base step of the induction, $t^0 = \GST$ eventually occurs.

Next, assume that $t^\ell$ occurs during the run.
If $t^{\ell+1}$ occurs, then we are done.

Assume by contradiction that $t^{\ell+1}$ does not occur, \ie, by the induction hypothesis some correct process entered $\rMax{t^{\ell}}$ but no correct process enters any round $r > \rMax{t^{\ell}}$.
By \ref{prop:sync1}, eventually at least $f+1$ correct processes enter \rMax{t^\ell}.
Denote this set of processes by $P$.
The timer function ensures that eventually every process in $P$ calls \wishToAdvance, so there are at least $f+1$ correct processes in \rMax{t^\ell} that call \wishToAdvance.
By \ref{prop:prog1}, eventually at least one correct process enters $\rMax{t^{\ell+1}} = \rMax{t^\ell} + 1$, a contradiction to the assumption that no correct process enters any round $r > \rMax{t^{\ell}}$.
\end{proof}
We prove the main theorem of this section:
\begin{theorem}
Using a synchronizer abstraction, \Cref{alg:timerLeader} implements a round synchronization module.
\end{theorem}
\begin{proof}
Since the relay function's outputs are identical among all correct processes and the leader local function outputs \newLeader{r,\relay{r,1}}, it is immediate that the leader agreement property (\Cref{prop:leaderAgreement}) is satisfied.

We now prove eventual round synchronization (\Cref{prop:roundClockEventualSync}).
Define $\leader{\ell} \triangleq \relay{\rMax{t^\ell},1}$.
By \Cref{lem:indefinitelyNewRounds}, $t^i$ occurs for all $i \geq 0$, and since the first relay for each round is randomly chosen, eventually, with probability $1$, there exists a $\ell \geq 0$ such that \leader{\ell} is a correct process.
Let us look at $\rMax{t^\ell}$, and denote $\widetilde{t} \triangleq t^\ell+c_1+\Delta$.

Recall that $t^\ell$ is the time when the first correct process enters \rMax{t^\ell}.
By \Cref{alg:timerLeader:advance} in \Cref{alg:timerLeader}, no correct process calls \wishToAdvance between $t^\ell$ and $\widetilde{t}$, and because of validity (\Cref{prop:validity}) no correct process enters any round $r > \rMax{t^\ell}$ until at least $\widetilde{t}$.
By using \Cref{observ:stabilization}, we can apply \ref{prop:sync2} for \rMax{t^\ell}, since by $t^\ell + c_1$ all correct processes enter \rMax{t^\ell}.

Thus, between $t^\ell$ and  $t^\ell + c_1$, all correct processes enter \rMax{t^\ell}.
Since no correct process calls \wishToAdvance until at least $\widetilde{t}$, this guarantees that all correct processes remain in \rMax{t^\ell} until $\widetilde{t} = t^\ell+c_1+\Delta$, so $t^\ell + c_1$ is a synchronization time (\Cref{def:syncTime}), as needed.
\end{proof}
\section{An Expected Linear Message Complexity and Constant Latency Synchronizer}
\label{sec:algorithm}

\begin{figure}[t]
	\centering
	\includegraphics[width=\textwidth]{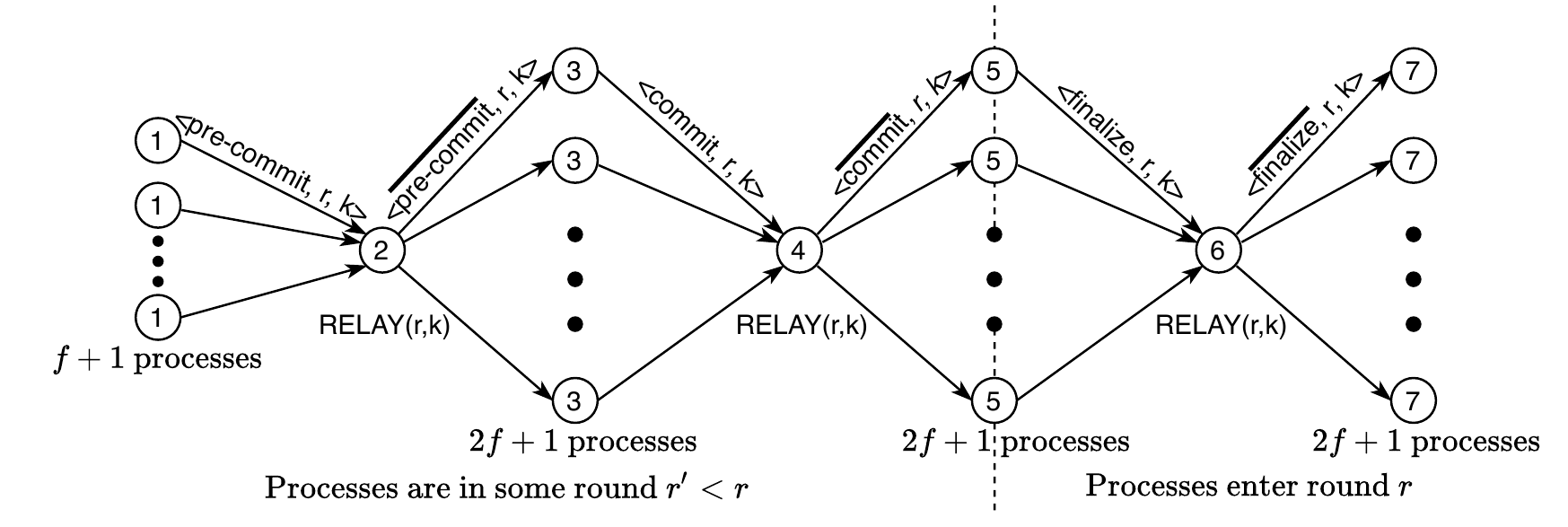}
	\caption{\textbf{The message flow of the algorithm.} {A process enters round~$r$ when it receives a \textoverline{commit} message for round $r$, the circled numbers represent the different stages of the algorithm (see~\Cref{alg}). Only the $2f+1$ correct processes are illustrated.}}
	\label{fig:message_flow}
\end{figure}

In this section we present a synchronizer abstraction algorithm with expected linear message complexity and constant latency in the Byzantine case.

We start by describing the main ideas used to lower the message complexity (while still guaranteeing constant latency) in \Cref{sec:algs:achievingLinearComplexity}.
We give a more in-depth description of the algorithm in \Cref{sec:algs:algDescription}, reason about the algorithm's correctness in \Cref{sec:algs:proofs}, and performance in \Cref{sec:algs:performance}.

\subsection{Achieving Linear Message Complexity} \label{sec:algs:achievingLinearComplexity}
The crux of the algorithm is a relay-based distribution of messages among processes.
A standard Byzantine broadcast system, which ensures that a message sent by a correct process is eventually delivered by all other correct processes, usually requires quadratic message complexity for each message disseminated.
This is because in Byzantine broadcast protocols such as Bracha's~\cite{bracha1987asynchronous}, when a correct process delivers a message, it also sends it to all the other processes, resulting in all-to-all communication for each delivered message.

In our algorithm, we instead use a single designated process as a relay. 
Processes send their messages to the relay, which aggregates messages from a number of processes, combines them into one message using a threshold signature, and broadcasts it to all the processes.
This mechanism reduces the total number of protocol messages from $O(n^2)$ to $O(n)$.

A difficulty arises if the relay is Byzantine.
We overcome this as follows: when a process \process{} sends a message to a relay, it expects a response from it within a certain time-bound.
If no timely response arrives, \process{} can deduce that either \GST has not occurred yet and the message to/from the relay is delayed, or it is after \GST and the relay is Byzantine.
In either case, after the allotted time passes, \process{} proceeds to send a message to a different relay, again waiting for the new relay to respond in a timely manner, and so on.
This mechanism uses the relay function described in \Cref{sec:model}. 
Once a correct relay is contacted, the algorithm makes progress.
In expectation, the number of consecutive Byzantine relays until a correct one is bounded by $3/2$, leading to expected constant latency and linear message complexity.
In the worst-case, each round has $f+1$ potential relays, guaranteeing that at least one of them is correct, which ensures liveness.

\subsection{Algorithm Description} \label{sec:algs:algDescription}
At a high level, the goal of the algorithm is to eventually enter all rounds during the run, and reach a synchronization time after \GST in every round $r$ where \relay{r,1} is a correct process.
If the relay is Byzantine, then the goal is to eventually move from $r$ to $r+1$.
The randomization of the relay function guarantees that in an infinite run there will be infinitely many rounds with a correct process as the first relay, guaranteeing an infinite number of synchronization times.

\textbf{Message flow of the algorithm.}
The algorithm is presented in~\Cref{alg}, and its message flow is depicted in \Cref{fig:message_flow}.
Protocol messages are signed and verified; for brevity, we omit the signatures and their verification from the algorithm description and pseudocode.

A process sends to \relay{r,k} messages of the form $\left\langle \text{message type}, r, k \right\rangle$, where $r$ and $k$ are natural numbers, and message type is one of the following: pre-commit, commit, or finalize.
The relay's messages to the processes are threshold signatures on an aggregation of the same messages, denoted $\TC{r,k}$, $\QC{r,k}$, and $\finalizeAck{r,k}$, respectively.
Each threshold signature is created using some number ($f+1$ or $2f+1$) of signatures.

When \wishToAdvance is signaled via the local timer function (see \Cref{sec:problemDef:localFuncs}) to indicate that it wants to move from round $r-1$ to round $r$, the process sends a pre-commit message to the relay (this is stage 1 of the algorithm).
Once $f+1$ processes indicate that they wish to move to round $r$, the relay broadcasts a \textoverline{pre-commit} message (stage 2).
The reason $f+1$ processes are needed to initiate the first stage of the algorithm is to ensure that there is at least one correct process among them, preventing Byzantine processes from causing correct ones to advance prematurely.
Any process receiving a relay's \textoverline{pre-commit} message in a round $r' < r$ joins in by sending a commit message for $r$ (stage 3).
Unlike in previous work such as Cogsworth~\cite{naor2019cogsworth}, \textoverline{pre-commit} messages are linked to a particular relay, and therefore, if the protocol times out and proceeds to the next relay, the new relay needs to collect $f+1$ pre-commits afresh.
This subtle difference prevents Byzantine relays from spuriously engaging in the protocol, which is crucial for avoiding the quadratic message complexity occurring in Cogsworth.

When $2f+1$ processes indicate that they commit to moving to $r$, the relay sends a \textoverline{commit} message (stage 4) and processes that receive it enter that round (stage 5).
Requiring $2f+1$ processes to commit to a round $r$ before entering it ensures that at least $f+1$ correct processes are aware of the intent to enter $r$.
This ensures that at least $f+1$ correct processes will eventually enter $r$, and those $f+1$ processes guarantee progress, as it is the minimal quorum required to initiate the stages of the algorithm to the next round, until a round with a first correct relay is reached and in that round a synchronization time will occur.

However, the algorithm for synchronizing for round $r$ does not end when a process receives a \textoverline{commit} message for~$r$.
Rather, a process that enters round $r$ sends a finalize message to help any lagging processes with the transition to round $r$.
Once $2f+1$ finalize messages are sent, the relay broadcasts a \textoverline{finalize} message (stage 6), and when a process receives it, it completes the algorithm for round $r$ (stage 7).
The finalization phase is needed to overcome cases of a Byzantine relay that does not send the \textoverline{commit} message to all the processes.

\textbf{Variables and timeouts.}
The variable \curr stores the current round a process is currently in which changes in stage 5, and \nextVar indicates to what round the process is attempting to enter.
The value of \nextVar becomes $\curr + 1$ when a process invokes \wishToAdvance, and it can become higher if the process learns (via a \textoverline{pre-commit}) of at least $f+1$ other processes that want to advance to a higher round than the one the process is currently in.

The timeouts at the bottom of the pseudocode dictate when a process moves to the next relay of a round.
When a process sends a message to a relay, it expects the relay to respond within $2\dissTime$, which is the upper bound of the round-trip time after \GST.
For example, if a process sends a message of round $r$ to \relay{r,k} at time $t$ and does not receive a response by $t+2\dissTime$, it sends the message to \relay{r, k+1}.
This continues up to \relay{r, f+1}, guaranteeing that at least one of the relays for round $r$ is correct.

Upon a timeout, a process sends a pre-commit message to the next relay in line, and once that relay gets f+1 such messages, it, too, can try to complete the stages of the protocol for the same round.
There is a tradeoff involved in choosing the timeout – a shorter timeout may cause a second relay to engage even when the first relay is correct, whereas a longer one delays progress in case of a Byzantine relay. Nevertheless, it is important to note that a process responds to all relays, so contacting the $(k+1)$-st relay for round $r$ does not in any way prevent the $k$-th one from making progress. Thus, while setting an aggressive timeout may cause the protocol to send more messages, it does not in any way hamper progress.
A process that partakes in the protocol to advance to round $r$ contacts a new relay every $2 \dissTime$ time for as long as it does not make progress in the phases of the algorithm for round $r$.
Since a process takes an expected $6\delta$ to complete the algorithm for round $r$, the process contacts $3$ relays in expectation.

The \tempRelay array holds the highest relay for each round the process sent a pre-commit message to.
For example, $\tempRelay[r] = k$ for $k > 1$ indicates that the process sent $\wish{r,1}, \ldots, \wish{r,k}$ messages to $\relay{r,1}, \ldots, \relay{r,k}$, respectively.
Note that a process sends a pre-commit message for round $r$ to \relay{r,1} when it first receives a \textoverline{pre-commit} message in stage 3, regardless of the relay it received the message from.
%In a similar way, a process sends a commit message for round $r$ to \relay{r,1} in stage 5.
This is to allow the first relay of round $r$ to complete the stages of the algorithm in case it is correct, and make sure that round synchronization will occur in round~$r$.
Note that the fact that some relay sends a message with a threshold signature does not ensure that that relay is correct, even if all the signatures used to create the threshold signature are from correct processes.
For example, a Byzantine relay can broadcast a message only to a subset of the correct processes.
Thus, to ensure liveness, processes must iterate through all $f+1$ relays of a round, starting from the first one, until progress is made.

We note that the \tempRelay array is introduced in the pseudocode for simplicity, but in a real implementation there is no need for an unbounded array to be stored in memory.
A process only sends messages to the relays of rounds stored in the  \curr and \nextVar variables, thus limiting the amount of memory needed for an actual implementation to a constant number of integers.

\textbf{Example.}
To clarify the need for the last phase of the algorithm (stages 6 and 7), consider the following scenario:
Suppose a set $P$ of $f+1$ correct processes are in round $r-1$ and invoke \wishToAdvance.
The remaining $f$ correct processes are in a round $r' < r-1$.
The processes in $P$ send a pre-commit message to \relay{r,1}, which is Byzantine.
The relay generates a threshold signature and sends a \textoverline{pre-commit} only to the processes in $P$, which respond with a commit message.
Now, \relay{r,1}, with the help of $f$ Byzantine processes, creates a \textoverline{commit} message for $r$, but sends it to only one correct process $\process{i}$ in $P$.
This results in a scenario where \process{i} is the only correct process in round $r$, while $f$ correct processes remain in round $r-1$ and continue to timeout and send pre-commit messages to the relays of round $r$.
Since a relay needs at least $f+1$ pre-commit messages to engage the stages of the algorithm, unless \process{i} continues to help the rest of the processes in $P$ by sending pre-commit messages, they might get stuck in round $r-1$.
Therefore, processes continue to timeout and send pre-commit messages in the previous round until they receive a \textoverline{finalize} message. 
Once a process in $r$ receives a \textoverline{finalize} message for $r$, it knows that there are at least $f+1$ correct processes in round $r$, and can stop sending pre-commit messages for $r$.
This is crucial for achieving the desired message complexity after GST.
These $f+1$ correct processes will eventually call \wishToAdvance and proceed to round $r+1$.

\makeatletter
\let\oldnl\nl% Store \nl in \oldnl
\newcommand{\nonl}{\renewcommand{\nl}{\let\nl\oldnl}}% Remove line number for one line
\makeatother
\begingroup
\singlespacing
\begin{algorithm*}[H] \smaller[1]
	\caption{\textbf{Synchronizer Algorithm}. The circles show the protocol's stages.}
	\label{alg}
	\SetAlgoNoEnd
	\DontPrintSemicolon
	\SetInd{0.4em}{0.4em}
	\KwInitialize(\unskip:\label{alg:relibra:initialize} )  
	{
		$\curr \gets 0$ \tcp*[h]{Processes begin their execution at round 0.} \;
		$\nextVar \gets 0$ \;
		$\forall i \in \mathbb{N} \colon \tempRelay[i] \gets 1$ \;
		$\finalized \gets \true$ \;
	}
	\setlength{\columnsep}{20pt}
	\BlankLine 	\BlankLine
	\begin{multicols*}{2}
	\nonl \underline{\textbf{Every process}}: \;
	\nonl \circled{1} \;
%	\vspace{-3pt}
	\nl \KwOn({\wishToAdvance signal:} \label{alg:relibra:wishToAdvance}) {
		\If(\tcp*[h]{old \newline \phantom{z} \qquad \qquad \qquad \qquad \qquad \qquad \quad \qquad  round}){$\curr < \nextVar$}
		{
			\Return \;
		}
		$\nextVar \gets \curr + 1$ \label{alg:line:advanceChangeCurr}\; 
		\send \newRound{\nextVar, 1} to \relay{\nextVar, 1} \label{alg:line:sendPC1}\;
	}

	\columnbreak
	\nonl \underline{\textbf{Relay ($\relay{r, k}$):}}\;
	\nonl \circled{2} \;
%	\vspace{-3pt}
	\nl \KwUpon({receiving the first valid $f + 1$ $\newRound{r, k}$ messages:} \label{alg:relibra:leaderReceiveTC}) {
		\broadcast $\TC{r, k}$
	}
	
	\end{multicols*}
	\BlankLine 
	\BlankLine
	\begin{multicols*}{2}
	\nonl \circled{3} \;
%	\vspace{-3pt}
	\nl \KwUpon({receiving the first valid $\TC{r, k}$ from $\relay{r,k}$:  \label{alg:relibra:TC}}) 
	{
		\If(\tcp*[h]{old round}){$r < \nextVar$} 
		{
			\Return \;
		}
		\If (\tcc*[h]{start participating \newline \phantom{z} \qquad \qquad \qquad \qquad  \qquad \quad in round $r$}){$r > \nextVar$ } 
		{
			$\nextVar \gets r$ \label{alg:line:updateCurrCommit}\;
		\send \newRound{r,1} to \relay{r,1} \label{alg:relibra:forwardTCto1} \label{alg:line:sendPC2} \label{alg:line:sendPCtoRelay1} \;	
		}
		\send \vote{r,k} to \relay{r,k} \;
	}

	\columnbreak
	\nonl \circled{4} \;
%	\vspace{-2pt}
	\nl \KwUpon({receiving the first valid $2f + 1$ $\vote{r, k}$ messages:} \label{alg:relibra:leaderReceiveQC}) {
	\broadcast $\QC{r, k}$ 
	}
	\end{multicols*}
	
	\BlankLine 	
	\BlankLine
	\begin{multicols*}{2}
	\nonl \circled{5} \;
%	\vspace{-5pt}	
	\nl \KwUpon({receiving the first valid \QC{r, k} from $\relay{r,k}$:  } \label{alg:relibra:QC})
	{
		\If(\tcp*[h]{old round}){$r < \curr$}
		{
			\Return \;
		}
		\If(\tcp*[h]{enter round $r$}){$r > \curr$ \label{alg:relibra:QC:ifLarger}}
		{
			$\curr \gets r$ \;	
			$\finalized \gets \false$ \;	
			\send \vote{r,1} to \relay{r,1} \label{alg:line:sendCommitToRelay1}\;
			$\proposeRound{r}$ \tcp*[h]{signal new round} \label{alg:relibra:advanceToRound} \;
		}
		\send \finalize{r,k} to \relay{r,k} \;
	}
		
	\columnbreak
	\nonl \circled{6} \;
%	\vspace{-2pt}	
	\nl \KwUpon({receiving the first valid $2f + 1$ $\finalize{r, k}$ messages:} \label{alg:relibra:leaderReceiveFinzalize}) {
	\broadcast $\finalizeAck{r, k}$ \label{alg:relibra:leaderMulticastQC}
	}
	\end{multicols*}
	\BlankLine
	\BlankLine
	\nonl \circled{7} \;	
%	\vspace{-5pt}		
	\nl \KwUpon({receiving the first valid \finalizeAck{r, k} from \newline$\relay{r,k}$: \label{alg:relibra:receiveFinalize}})
	{
		\If{$r = \curr$}
		{
			$\finalized \gets \true$ \;
		}
	}	
	\BlankLine
	\BlankLine
	\nonl \underline{\textbf{Timeouts (for every process):}}\;
	\KwOn({pre-commit and commit timeouts: \commentx{Every $2\dissTime$ from last sending pre-commit or commit messages and not receiving the matching \textoverline{pre-commit} or \textoverline{commit}}} \label{alg:relibra:attemptedTC})
	{
		\If{$\tempRelay[\nextVar] < f+1$  \label{alg:relibra:timeoutIf}}
		{
			$\tempRelay[\nextVar] \gets \tempRelay[\nextVar] + 1 $ \;
			\send $\newRound{\nextVar, \tempRelay[\nextVar]}$ to $\relay{\nextVar, \tempRelay[\nextVar]}$ \label{alg:line:sendPC3}\;
		}
	}
	
	\KwOn({finalize timeout: \tcp*[h]{Every $2 \dissTime$ from last sending finalize and not receiving the matching \textoverline{finalize}} \label{alg:relibra:timeoutFinalize}})
	{
		\If{ $\finalized = \false$ \KwAnd $\tempRelay[\curr] < f+1$}
		{
			$\tempRelay[\curr] \gets \tempRelay[\curr] + 1 $ \;
				\send $\newRound{\curr, \tempRelay[\curr]}$ to $\relay{\curr, \tempRelay[\curr]}$ \label{alg:line:sendPC4}\;
		}
	}	
\end{algorithm*}
\endgroup

\subsection{Correctness} \label{sec:algs:proofs}
Next, we prove that the algorithm satisfies the properties of a synchronizer, as defined in~\Cref{sec:problemDef:synchronizer}. 

\begin{restatable}{lemma}{monotonicalRounds} \label{lem:monotonicalRounds}
\Cref{alg} satisfies monotonically increasing rounds (\Cref{prop:monotonicRounds}).
\end{restatable}
\begin{proof}
	The algorithm signals $\proposeRound{r}$ in stage 5 only when it receives a \textoverline{commit} message for a round~$r$ that is larger than the one the process is currently in.
	Therefore, a process enters rounds in monotonically increasing order.
\end{proof}

\begin{restatable}{lemma}{lemmaValidity} \label{proposition:mustCallAdvance}
\Cref{alg} satisfies round validity (\Cref{prop:validity}).
\end{restatable}
\begin{proof}
	A correct process enters round $r$ when it is in a round $r' < r$ and receives a \textoverline{commit} message for $r$.
	A \textoverline{commit} message is a threshold signature of $(2f+1)$-of-$n$ commit messages, meaning at least $f+1$ are from correct processes.
	A correct process sends a commit message for round $r$ when it receives a \textoverline{pre-commit} message for $r$.
	A \textoverline{pre-commit} message is a threshold signature of $(f+1)$-of-$n$ pre-commit messages, meaning at least one correct process sent a pre-commit message for round $r$.
	
	Denote \process{i} as the first correct process that sends a pre-commit message for $r$ during the run.
	A correct process only sends a pre-commit for $r$ (in Lines \ref{alg:line:sendPC1}, \ref{alg:line:sendPC2}, \ref{alg:line:sendPC3}, and~\ref{alg:line:sendPC4}) when its \nextVar or \curr variables hold $r$.
	\nextVar changes in one of two places – \Cref{alg:line:advanceChangeCurr} when a process calls \wishToAdvance, and \Cref{alg:line:updateCurrCommit} on receiving a valid  \textoverline{pre-commit} for $r$.
	\curr changes on receiving a valid \textoverline{commit} for $r$.
	Because no \textoverline{pre-commit} or \textoverline{commit} message can be sent for round $r$  before at least one correct process sends a pre-commit for $r$, then \process{i} must have sent its pre-commit message for round $r$ when it changed its \nextVar in \Cref{alg:line:advanceChangeCurr}, \ie, on executing \wishToAdvance.
\end{proof}
\begin{proposition} \label{proposition:finalizeMessage}
	If a correct process receives a \textoverline{finalize} for round $r$ at time $t$, then at least $f+1$ correct processes entered round $r$ by $t$.
\end{proposition}
\begin{proof}
Let $t$ be a time in which a correct process received a \textoverline{finalize} message for round $r$.
This message is a threshold signature of $(2f+1)$-of-$n$ finalize messages, of which at least $f+1$ originated from correct processes.
A correct process only sends a finalize message for $r$ if it receives a \textoverline{commit} message for $r$, which means that it is already in round $r$ by time $t$.
\end{proof}
\begin{lemma} \label{claim:lumierePropSynchronization}
\Cref{alg} satisfies stabilization (\Cref{prop:synchronizerSyncronization})  with $c_1 = 4\dissTime$.
\end{lemma}
\begin{proof}
Let $t$ be a point in time during the execution and $r = \rMax{t}$.
Let \process{i} be the first correct process that enters round $r$ at time $t_0$.
Such a process exists by the definition of \rMax{t}.
\process{i} is at round~$r$, so it received a \textoverline{commit} message for round $r$.
A \textoverline{commit} message is a threshold signature of $(2f+1)$-of-$n$ commit messages, at least $f+1$ of which were sent by correct processes.
Denote by $S$ the set of correct processes whose signatures on commit messages are included in the \textoverline{commit} message \process{i} received.
The processes in $S$ are either in round $r$ at time $t$ or in smaller rounds $r' < r$.
	
We now prove the two sub-properties of  \Cref{prop:synchronizerSyncronization}:

\textbf{\ref{prop:sync1}.}
If some correct process receives \textoverline{finalize} for round $r$, by \Cref{proposition:finalizeMessage}, there are at least $f+1$ correct processes in $r$ and we are done.

Assume no correct process receives \textoverline{finalize}.
Then, the processes in $S$ continue to timeout and send pre-commit messages for round $r$ to the relays of $r$.
This guarantees that eventually, a correct relay for $r$ receives at least $f+1$ pre-commit messages, as Property \ref{prop:relayFunc:differentProcesses} of the relay function ensures $f+1$ different relays for each round.
This relay eventually completes the stages of the algorithm, allowing all correct processes to advance to round $r$.

\textbf{\ref{prop:sync2}.} Because \process{i} receives \textoverline{commit} for round $r$ at time $t_0 \geq \GST$, as argued above, $f+1$ correct processes have sent a commit message for round $r$ by time $t_0$.
Because a process sends pre-commit to \relay{r,1} before sending a commit to any relay for round $r$ (Lines \ref{alg:line:sendPC1} or \ref{alg:line:sendPCtoRelay1}), these messages, too, are sent by time $t_0$.
Therefore, by time $t_0+\delta$, \relay{r,1} receives $f+1$ pre-commit messages and sends a \textoverline{pre-commit} message to all processes.
By $t_0+2\dissTime$ all the correct processes receive the \textoverline{pre-commit} message sent from the first relay, by $t_0 +3\dissTime$ the relay receives $2f+1$ commit messages (along with any process that already entered $r$, \Cref{alg:line:sendCommitToRelay1}), and by $t_2 \leq t_0 + 4\dissTime$ all the correct processes receive the \textoverline{commit} message and enter round $r$.
\end{proof}
\begin{restatable}{lemma}{lumierePropProgress} \label{claim:lumierePropProgress}
	\Cref{alg} satisfies progress (\Cref{prop:synchronizerProgress}) with $c_2 = 4\dissTime$.
\end{restatable}
\begin{proof}
	Let $t$ be a certain point in time during the execution, $r = \rMax{t}$, and assume that by some time $t_0$, at least $f+1$ correct processes call \wishToAdvance while in round $r$ and not while in any round $r' > r$.
	By \Cref{proposition:mustCallAdvance}, no correct process enters any round $r'' > r+1$, \ie, the above group of at least $f+1$ processes can eventually be in either round $r$ or $r+1$ (by \Cref{prop:monotonicRounds}, rounds are monotonically increasing so they cannot be in any round lower than round~$r$).
	Denote this group of processes by $S$.
	
	We now prove the two sub-properties of \Cref{prop:synchronizerProgress}:
	
	\textbf{\ref{prop:prog1}.} Once some correct process receives \textoverline{finalize} for round $r+1$, by \Cref{proposition:finalizeMessage}, there are at least $f+1$ correct processes in $r$ and we are done.
	
	Assume no correct process receives \textoverline{finalize}.
	Then, the processes in $S$ continue to timeout and send pre-commit messages for round $r+1$ to the relays of $r+1$.
	This can either be done if the process in $S$ is in round $r$ and still did not receive a \textoverline{commit} message for $r+1$, or if the process is already in $r+1$ but did not receive \textoverline{finalize} message for $r+1$ and continues to timeout (\Cref{alg:relibra:timeoutFinalize}).
	This guarantees that eventually, a correct relay for $r+1$ receives at least $f+1$ pre-commit messages, as Property \ref{prop:relayFunc:differentProcesses} of the relay function ensures $f+1$ different relays for each round.
	This relay eventually completes the algorithm, allowing all correct processes to advance to round $r+1$.

	\textbf{\ref{prop:prog2}.} 
	Assume \relay{r+1,1} is correct.
	By $t_0 \geq \GST$ at least $f+1$ correct processes call \wishToAdvance while in round $r$, therefore, by $t_0+\dissTime$ \relay{r+1,1} receives enough pre-commit messages for round $r+1$ to engage the first phase of the algorithm, and by $t_0+2\dissTime$ all the correct processes receive the \textoverline{pre-commit} message, and by $t_0+3\dissTime$ the relay receives $2f+1$ commit messages.
	Even if a process enters round $r+1$ thorough a different relay than \relay{r+1,1}, it still sends to \relay{r+1,1} the commit message for round $r+1$ (in \Cref{alg:line:sendCommitToRelay1}).
	Thus, the first relay has enough commit messages, and by $t_2 \leq t_0+4\dissTime$ all the correct processes receive the \textoverline{commit} message and enter round $r+1$.
\end{proof}

The following theorem follows directly from Lemmas \ref{lem:monotonicalRounds}, \ref{proposition:mustCallAdvance}, \ref{claim:lumierePropSynchronization}, and \ref{claim:lumierePropProgress}.
\begin{theorem}
	\Cref{alg} satisfies the synchronizer abstraction.
\end{theorem}

\subsection{Performance: Latency and Message Complexity} \label{sec:algs:performance}
We prove that our algorithm has expected constant latency and linear message complexity.

The proof of the next proposition uses hypergeometrical distribution.
A random variable $Z$ which is hypergeometrically distributed with the parameters $Z \thicksim \textit{HG}(N,D,K)$ describes the number of successes (random draws where the drawn object has a specific feature) in a finite population of $N$ objects, of which $D$ objects have the feature.
$K$ is the total number of draws, which are done \emph{without} replacement of the objects.
The pmf of this distribution is
\begin{equation*}
\Pr \left[ Z = z \right] = \frac{\binom{D}{z}\binom{N-D}{K-z}}{\binom{N}{K}}.
\end{equation*}
\begin{restatable}{proposition}{expectedByzantineRelays} \label{prop:expectedByzantineRelays}
	For any round $r$, let $X_r$ be the number of consecutive Byzantine relays until the first correct relay.
	Then, $\forall r \colon \mathbb{E} \left[ X_r \right] \leq 3/2$.
\end{restatable}
\begin{proof}
	Since the relays for each round $r$ are randomly chosen regardless of $r$, the expectation of $X_r$ is the same for all $r$, and for brevity, we omit the round notation.
	
	Based on the definition of $X$, and the fact that relays for each round are randomly chosen without replacement:
	\begin{equation*}
	\Pr \left[ X = x \right] = \Pr \left[ Y = x-1 \right] \cdot \frac{x}{n - \left( x - 1 \right)},
	\end{equation*}
	where $Y$ is a hypergeometrical distribution with the following parameters: $Y \thicksim \textit{HG}(n,f,x-1)$. 
	
	Thus, for any $n$ and $f < n/3$:
	\begin{align*}
	& \Pr \left[ X  = x \right] = \Pr \left[Y = x-1 \right] \cdot \frac{n-f}{n-\left(x-1\right)} = \frac{\binom{f}{x-1} \binom{n-f}{x-1 - \left( x-1\right)}}{\binom{n}{x-1}} \cdot \frac{n-f}{n-\left(x-1\right)} \\
	&= \frac{f! \left( n-x+1 \right)!}{n! \left( f-x+1 \right)!} \cdot \frac{n-f}{n-\left(x-1\right)}
	\end{align*}
	And in expectation
	\begin{align*}
	& \mathbb{E} \left[ X \right] = \sum_{x=1}^{f+1} x \Pr \left[ X = x \right] = \sum_{x=1}^{f+1} x \cdot \frac{f!(n-x+1)!}{n!(f-x+1)!} \cdot \frac{n-f}{n-x+1} = \frac{n+1}{n-f+1} \leq \frac{3}{2}.
	\end{align*} 
\end{proof}

\begin{restatable}{lemma}{expectationOfStabilization} \label{lemma:expectationOfStabilization}
For any $t \geq \GST$ let $t_0$ be the first time during the run where a correct process enters \rMax{t}.
There exists a time $t_1 \geq t_0$ such that up to $t_1$ either \textit{(i)}~at least $f+1$ correct processes are in \rMax{t} \textbf{or} \textit{(ii)}~a correct process enters a round $r > \rMax{t}$;
and $\mathbb{E} \left[ t_1 - \max\{t_0,\GST\} \right] \leq \frac{3}{2} \cdot 6\dissTime$.
\end{restatable}
\begin{proof}
	The first correct process that enters $\rMax{t}$ does so because it receives a \textoverline{commit} message at $t_0$.
	This message contains signatures of $2f+1$ processes, of which at least $f+1$ are correct.
	Denote the set of those correct processes as $S$.
	The processes in $S$ have all sent pre-commit messages to \relay{\rMax{t},1} before sending their signed commit messages, no later than $t_0$.
	If no correct process enters a round $r > \rMax{t}$, then each of the processes in $S$ waits up to $2\dissTime$ for a \textoverline{pre-commit} from \relay{r,1}, and if it is received, waits at most $2\dissTime$ for a \textoverline{commit}, and if it is received, waits again at most $2\dissTime$ for a \textoverline{finalize}. 
	
	Thus, within $6\delta$, one of the following occurs first:
	\begin{enumerate}
		\item \label{proof:expectedStabilizationTIme:1} At least one process in $S$ receives a \textoverline{finalize} message, and by \Cref{claim:lumierePropSynchronization} and \Cref{observ:stabilization}, $t_1$ occurs by case~\textit{(i)}.
		\item \label{proof:expectedStabilizationTIme:2} The processes in $S$ timeout and send pre-commit to the next relay. 
		This is repeated until a correct relay completes the stages of the algorithm and $t_1$ occurs by case~\textit{(i)}.
		\item \label{proof:expectedStabilizationTIme:3} A correct process enters round $\rMax{t} +1$ before (\ref{proof:expectedStabilizationTIme:1}) or (\ref{proof:expectedStabilizationTIme:2}) occur, resulting in $t_1$ occurring by case~\textit{(ii)}.
	\end{enumerate}
	Because (\ref{proof:expectedStabilizationTIme:2}) only occurs if the relay is Byzantine, and by \Cref{prop:expectedByzantineRelays} the expected number of relays until a correct one is bounded by $3/2$, we get that within expected $3/2$ iterations, each taking at most $6\delta$, $t_1$ occurs by (\ref{proof:expectedStabilizationTIme:1}) or (\ref{proof:expectedStabilizationTIme:3}), and the upper bound for the expectation is $\mathbb{E} \left[ t_1 - \max \left\{ t_0, \GST \right\} \right] \leq \frac{3}{2} \cdot 6\dissTime = 9\dissTime$.	
\end{proof}

\begin{restatable}{lemma}{lemmaExpectationOfProgress} \label{lemma:expectationOfProgress}
	For any $t \geq \GST$ let $t_0$ be the first time when $f+1$ correct processes call \wishToAdvance while in round $\rMax{t}$.
	There exists a time $t_1 \geq t_0$ such that there is at least one correct process in $\rMax{t}+1$ and $\mathbb{E} \left[ t_1 - \max \left\{ t_0, \GST \right\} \right] \leq \frac{3}{2} \cdot 6\dissTime$.
\end{restatable}
\begin{proof}
	Let $S$ be the set of correct processes that call \wishToAdvance in round $\rMax{t}$ by $t_0$.
	The reasoning for this property is similar to the one in the proof of \Cref{lemma:expectationOfStabilization}, \ie, since the processes in $S$ timeout at most every $6\dissTime$ until they move to the next relay of round $\rMax{t}+1$, and there are enough processes in $S$ to eventually send $f+1$ pre-commit messages to a correct relay, and once this relay is found it can complete the stages of the algorithm.
	In expectation, the time until $t_1$ depends on the number of Byzantine relays until a correct one.
	By \Cref{prop:expectedByzantineRelays} it is bounded by $3/2$, and therefore $\mathbb{E} \left[ t_1 - \max \left\{ t_0, \GST \right\} \right] \leq \frac{3}{2} \cdot 6\dissTime = 9\dissTime$.
\end{proof}

\begin{theorem}
The synchronizer algorithm along with a Timer local function (as defined in \Cref{sec:decomposition}) achieves expected constant latency and linear message complexity.
\end{theorem}
\begin{proof}
The latency for our algorithm is based on the definition in \Cref{sec:problemDef:performance}.
We go over all possible states after \GST the correct processes in our algorithm can be in, and look at the expected latency until the synchronization time.	
Let $t^0 = \GST$ and for all $\ell \geq 1$ let $t^\ell$ represent the first time after $t^{\ell-1}$ that a correct process enters a new maximum round.
By \Cref{lem:indefinitelyNewRounds}, in an infinite run, $t^\ell$ eventually occurs for any $\ell \geq 0$.
For any time $t \geq \GST$ during the run, let $\syncTime{t}$ be the first time after $t$ until a synchronization time (\Cref{def:syncTime}).
To calculate the expected latency of our algorithm, we need to show that for any $t \geq \GST$,  $E_1 \triangleq \mathbb{E} \left[ \syncTime{t} - t \right] \leq O(\dissTime)$.

Denote $E_2$ as the expected time from any time $t \geq \GST$ until the next $t^\ell$, \ie, for any $l \geq 0$ and $t$, $E_2 \triangleq \mathbb{E} \left[ \min_{t^\ell \geq t} \left\{ t^\ell \right\} -t \right]$
and $E_3 \triangleq \mathbb{E} \left[ t^{\ell+1} - t^{\ell} \right]$. 
If \relay{\rMax{t^\ell},1} is correct, then based on \ref{prop:prog2}, by $t^\ell + 4\dissTime$ all the correct processes enter \rMax{t^\ell}.
Therefore:
\newcommand\numberthis{\addtocounter{equation}{1}\tag{\theequation}}
\begin{align*}
 E_1 & \leq E_2 +  \underbrace{\frac{n-f}{n}}_{\substack{\text{Probability that} \\ \text{\relay{\rMax{t^\ell},1}} \\ \text{is correct}}} \cdot \underbrace{4\dissTime}_{\substack{ \text{The maximum time for} \\ \text{all correct processes} \\ \text{to enter a round} \\ \text{(\Cref{claim:lumierePropSynchronization})} }} + \underbrace{\frac{f}{n}}_{\substack{\text{Probability that} \\ \text{\relay{\rMax{t^\ell},1}} \\ \text{is Byzantine}}} \cdot \underbrace{  \mathbb{E} \left[ \syncTime{t^{\ell+1}} - t^\ell \right] }_{ \substack{ \text{The expected time until all} \\ \text{correct processes} \\ \text{enter \rMax{\syncTime{t^{\ell+1}} }}}} =  \\
 & = E_2 +  \frac{n-f}{n} \cdot 4\dissTime + \frac{f}{n} \cdot \big( E_1 + \underbrace{\mathbb{E} \left[ t^{\ell+1} - t^\ell \right]}_{=E_3} \big) \\
 &\Rightarrow  E_1 \leq \frac{n}{n-f} \left( E_2 + \frac{n-f}{n}\cdot 4\dissTime + \frac{f}{n} \cdot E_3 \right). \numberthis \label{eqn}
\end{align*}
Assuming that once a correct process enters a new round, the timer calls \wishToAdvance within $4\dissTime + \Delta$, the expected time between $t^{\ell}$ and $t^{\ell+1}$ can be bounded as follows:
\vspace{-0.5em}
\begin{align*}
E_3 = \mathbb{E} \left[ t^{\ell+1} - t^\ell \right] & \leq \underbrace{ \mathbb{E} \left[ {  \substack{  \text{Time from $t^\ell$ until} \\ \text{ at least $f+1$ correct} \\ \text{processes enter \rMax{t^{\ell}}} \\ \text{or $t^{\ell +1}$ occurs}}}\right] }_{\text{\Cref{lemma:expectationOfStabilization}}} + \underbrace{4\dissTime + \Delta}_{\substack{\text{Time until at least $f+1$} \\ \text{correct processes call} \\ \text{\wishToAdvance in \rMax{t^\ell}}}} + \underbrace{ \mathbb{E} \left[ {  \substack{  \text{Time from the first} \\ \text{ time $f+1$ correct} \\ \text{processes call \wishToAdvance} \\ \text{in \rMax{t^{\ell}} and} \\ \text{until $t^{\ell+1}$ occurs}}}\right] }_{\text{\Cref{lemma:expectationOfProgress}}}  \\
& \leq \frac{3}{2} \cdot 6\dissTime + 4\dissTime + \Delta + \frac{3}{2} \cdot 6\dissTime = 22\dissTime + \Delta.
\end{align*}

\vspace{-0.5em}
The calculation of $E_3$ proves that in expectation, the time between any $t^\ell$ and $t^{\ell+1}$ is expected constant, assuming $\Delta$ is constant.
Therefore, $E_2$ is also expected constant.

To conclude, we proved that $E_2 \leq O(\dissTime)$ and $E_3 \leq O(\dissTime)$, and by \Cref{eqn}, $E_1 \leq O(\dissTime)$, as needed to prove expected constant latency.

For the message complexity of the synchronizer, note that since the expected time between two occurrences of round synchronization is expected constant, the message complexity is expected linear.
This is because for a given round the number of consecutive Byzantine relays until a correct one is expected constant, and in the algorithm, every process sends one message to the relay in each stage of the algorithm, and the relay responds with one message to all the processes.
Even if a process contacts more than one relay per round, it still contacts an expected constant number of relays, and therefore this does not hamper the asymptotic linear message complexity.
\end{proof}

\subsection{Relaxed Model} \label{sec:algs:relaxedModel}
As part of the model in~\Cref{sec:model} we assumed that the adversary is oblivious.
If the adversary is \emph{strong}, and knows the randomness $\mathcal{R}$ before choosing which processes to corrupt, a worst-case bound is tantamount to a deterministic one, because it holds for all coin flips.
Therefore, we cannot hope to get a linear message complexity for the worst-case.
Nevertheless, in a run with infinitely many round synchronization events, we can bound the \emph{average-case} expected latency and message complexity by considering the limit of the average latency and message complexity on prefixes of length $t$ of the run as $t$ tends to infinity.

Thus, a strong adversary who is aware of the relay function, can choose to corrupt, e.g., the first $f$ relays of some round $r$, causing that round to have linear latency and quadratic message complexity.
But since the adversary is static, it has to corrupt the same processes in all rounds, and by property \ref{prop:relayFunc:random}, this does not impact the average-case performance.

\section{Related Work} \label{sec:relatedWork}
Algorithms for the eventual synchrony model almost invariably use the notion of round or views~\cite{lamport2001paxos,oki1988viewstamped,keidar1996corel,birman1987exploiting}. 
A number of works have suggested frameworks and mechanisms for round synchronization in the benign case~\cite{awerbuch1985complexity,ford2019threshold,keidar2006timeliness,keidar2008howToChoose,gafni1998RRFD}.
For example, Awerbuch introduced synchronizers~\cite{awerbuch1985complexity} for failure-free networks.
TLC~\cite{ford2019threshold} places a barrier on round advancement, so that processes enter round $r+1$ only after a threshold of the processes entered round $r$.
Frameworks like RRFD~\cite{gafni1998RRFD} and GIRAF~\cite{keidar2006timeliness,keidar2008howToChoose} create a round-based structure for eventually synchronous and failure-detector based algorithms.

A related concurrent work due to Bravo et al.~\cite{bravo2020making} also tackles the liveness of consensus protocols, and creates a general framework to abstract the liveness part of consensus protocols.
They show that some protocols such as PBFT~\cite{castro1999practical} and HotStuff~\cite{yin2019hotstuff} can use this framework.
They also provide an algorithm for round synchronization that, starting from some round $r$, synchronizes all rounds $r' \geq r$ (even rounds with a Byzantine leader, in which decisions are not made), but unlike our algorithm, it  requires a quadratic communication cost per synchronization event.
They assume a relaxed network  model compared to us, where before \GST messages might be lost.

Several algorithms include two modes of operation: a normal mode where the leader is correct incurring linear message complexity, and a recovery mode when the leader is faulty and needs to be replaced incurring quadratic or higher message complexity.
For example, PABC~\cite{ramasamy2005parsimonious} achieves amortized linear  message complexity in an asynchronous atomic broadcast protocol, and Zyzzyva~\cite{kotla2007zyzzyva} implements a linear fast-track in an SMR algorithm.

Randomization is often used to solve consensus in asynchronous networks to circumvent the seminal FLP result~\cite{fischer1985impossibility}.
VABA~\cite{abraham2019asymptotically} is the first multi-value asynchronous consensus algorithm that achieves an expected quadratic message complexity against a strong adaptive adversary, and other works in the asynchronous model~\cite{cohen2020coinincidence,blum2020byzantine} achieve expected sub-quadratic communication complexity under various assumptions, but do not achieve $O(n)$ communication complexity.
In the context of Byzantine SMR, HotStuff~\cite{yin2019hotstuff} specified the round synchronization conditions needed for their algorithm, and abstracted it into a module that was left unspecified.
Our work provides the round synchronization they require.

Our algorithm builds on ideas presented in Cogsworth~\cite{naor2019cogsworth}, but Cogsworth achieved expected linear message complexity only in the benign case, whereas in the Byzantine case its message complexity was still expected quadratic. 
In Cogsworth, it is enough that the first relay of a round is Byzantine to create a run with a quadratic number of messages to synchronize for that round.
Since the probability that the first relay is Byzantine is constant, \ie, $f/n$, the overall message complexity under Byzantine failures is expected quadratic.

To reduce the expected message complexity to linear, we modified Cogsworth in a number of ways, including adding another phase to the algorithm, signing each message from a process to a relay with the relay it is intended for, and adding a ``helping'' mechanism to help processes ``catch-up'' to the latest round.
By incorporating these ideas into our algorithm, we managed to bring the expected message complexity down to linear.

\section{Conclusion} \label{sec:conclusion}
We presented an algorithm that reduces the expected message complexity of round synchronization to linear with an expected constant latency.
Combined with algorithms like HotStuff, this yields, for the first time, Byzantine SMR with the same asymptotic performance, as round synchronization is the ``bottleneck'' in previous Byzantine SMR algorithms.
While we achieve only expected sub-quadratic complexity, we note that achieving the same complexity in the worst-case is known to be impossible~\cite{dolev1985bounds}, and so cannot be improved.

%% Please use bibtex, 

\bibliography{references}
\clearpage

\end{document}